\documentclass[acmsmall]{acmart}
\renewcommand\footnotetextcopyrightpermission[1]{} 
\usepackage{csquotes} 
\usepackage{enumerate} 
\usepackage{algorithm,bbm}
\usepackage[noend]{algpseudocode}
\usepackage{amsmath,amsthm} 
\usepackage{amsfonts}
\usepackage{graphics,graphicx}

\newcommand{\din}{\delta^{\mathrm{in}}}
\newcommand{\dout}{\delta^{\mathrm{out}}}
\newcommand{\DS}{\text{DS}}
\newcommand{\T}{\mathcal{T}}

\usepackage{mathtools}
\DeclareMathOperator\cost{cost}
\DeclarePairedDelimiter{\floor}{\lfloor}{\rfloor}
\DeclarePairedDelimiter{\ceil}{\lceil}{\rceil}

\newtheorem{theorem}{Theorem}[section]
\newtheorem{lemma}[theorem]{Lemma}

\usepackage{tikz}
\usetikzlibrary{arrows,decorations.markings, shapes.misc, shapes.multipart, positioning}

\usepackage{hyperref}
\AtBeginDocument{%
  \providecommand\BibTeX{{%
    \normalfont B\kern-0.5em{\scshape i\kern-0.25em b}\kern-0.8em\TeX}}}

\setcopyright{none}

\begin{document}

\title{Time- and space-optimal algorithm for the many-visits TSP}
\titlenote{A preliminary version of this paper appeared in the ACM-SIAM Symposium on Discrete Algorithms (SODA) 2019.}

\author{Andr{\'e} Berger}
\affiliation{%
  \institution{Maastricht University}
  \department{Department of Quantitative Economics}
  \city{Maastricht}
  \country{The Netherlands}}
\email{a.berger@maastrichtuniversity.nl}

\author{L{\'a}szl\'{o} Kozma}
\email{laszlo.kozma@fu-berlin.de}
\affiliation{%
  \institution{Freie Universit\"at Berlin}
  \department{Institute of Computer Science}
  \city{Berlin}
  \country{Germany}
}

\author{Matthias Mnich}
\affiliation{%
  \institution{Technische Universit{\"a}t Hamburg}
  \department{Institute for Algorithms and Complexity}
  \city{Hamburg}
  \country{Germany}}
\email{mmnich@tuhh.de}

\author{Roland Vincze}
  \affiliation{%
  \institution{Technische Universit{\"a}t Hamburg}
  \department{Institute for Algorithms and Complexity}
  \city{Hamburg}
  \country{Germany}}
\email{roland.vincze@tuhh.de}


\begin{abstract}
  The many-visits traveling salesperson problem (MV-TSP) asks for an optimal tour of~$n$ cities that visits each city $c$ a prescribed number $k_c$ of times.
  Travel costs may be asymmetric, and visiting a city twice in a row may incur a non-zero cost.
  The MV-TSP problem finds applications in scheduling, geometric approximation, and Hamiltonicity of certain graph families.
  
  The fastest known algorithm for MV-TSP is due to Cosmadakis and Papadimitriou (SICOMP, 1984).
  It runs in time $n^{O(n)} + O(n^3 \log \sum_c k_c )$ and requires $n^{\Theta(n)}$ space.
  An interesting feature of the Cosmadakis-Papadimitriou algorithm is its \emph{logarithmic} dependence on the total length $\sum_c k_c$ of the tour, allowing the algorithm to handle instances with very long tours. 
  The \emph{superexponential} dependence on the number of cities in both the time and space complexity, however, renders the algorithm impractical for all but the narrowest range of this parameter.

  In this paper we improve upon the Cosmadakis-Papadimitriou algorithm, giving an MV-TSP algorithm that runs in 
time $2^{O(n)}$, i.e.\ \emph{single-exponential} in the number of cities, using \emph{polynomial} space.
  The space requirement of our algorithm is (essentially) the size of the output, and assuming the Exponential-Time Hypothesis (ETH), the problem cannot be solved in time $2^{o(n)}$. 
  Our algorithm is deterministic, and arguably both simpler and easier to analyse than the original approach of Cosmadakis and Papadimitriou.
  It involves an optimization over directed spanning trees and a recursive, centroid-based decomposition of trees.
\end{abstract}


\maketitle

\thispagestyle{empty}
\pagestyle{empty}

\section{Introduction}
\label{sec:introduction}
The traveling salesperson problem (TSP) is one of the cornerstones of combinatorial optimization, with origins going back (at least) to the $19$th century work of Hamilton. For surveys on the rich history, variants, and current status of TSP we refer to the dedicated books~\cite{Lawler1985,GutinEtAl2002,Cook2011,ApplegateEtAl2006}.
In the standard TSP, given $n$ cities and their pairwise distances, we seek a tour of minimum total distance that visits each city.
If the distances obey the triangle inequality, then an optimal tour necessarily visits each city \emph{exactly once} (apart from returning to the starting city in the end).
In the general case of the TSP with arbitrary distances, the optimal tour may visit a city multiple times.
Instances with non-metric distances arise from various applications that are modeled by the TSP, e.g.\ from scheduling problems. 

To date, the fastest known exact algorithms for TSP (both in the metric and non-metric cases) are due to Bellman~\cite{Bellman1962} and Held and Karp~\cite{HeldKarp1962}, running in time $O(2^n n^2)$ for $n$-city instances; both algorithms also require space $\Omega(2^n)$.

In this paper we study the general problem where each city has to be visited \emph{exactly} a given number of times.
More precisely, we are given a set $V$ of $n$ vertices, with pairwise distances (or costs) $d_{ij} \in \mathbb{N} \cup \{\infty \}$, for all $i,j \in V$.
No further assumptions are made on the values $d_{ij}$, in particular, they may be asymmetric, i.e.\ $d_{ij}$ may not equal $d_{ji}$, and the cost $d_{ii}$ of a self-loop may be non-zero.
Also given are integers $k_i \geq 1$ for $i \in V$, which we refer to as \emph{multiplicities}.
A valid tour of length $k$ is a sequence $(x_1, \dots, x_k) \in V^k$, where $k=\sum_{i \in V}{k_i}$, such that each $i \in V$ appears in the sequence exactly $k_i$ times.
The cost of the tour is $\sum_{i=1}^{k-1}{d_{x_i,x_{i+1}}} + d_{x_k,x_1}$. Our goal is to find a valid tour with minimum cost. 

The problem is known as the \emph{many-visits TSP} (MV-TSP).
As an alternative name, \emph{high-multiplicity TSP} also appears in the literature.
When $k_i = 1$ for all $i \in V$, MV-TSP includes the standard metric TSP as a special case and thus, unless $\mathsf{P} = \mathsf{NP}$, it cannot be solved in polynomial time.
Nonetheless, as discussed later, the problem can be solved efficiently when the number of cities (vertices) $n$ is small, even if the length $k$ of the tour is very large (possibly exponential in $n$).

As a natural TSP-generalization, MV-TSP is a fundamental problem of independent interest.
In addition, MV-TSP has proved to be useful for modeling other problems, particularly in scheduling~\cite{Psaraftis1980, HochbaumShamir1991,BraunerEtAl2005,vanderVeenZhang1996}.
Suppose there are $k$ jobs of $n$ different types to be executed on a single, universal machine.
Processing a job, as well as switching to another type of job come with certain costs, and the goal is to find the sequence of jobs with minimal total cost.
Modeling this problem as a MV-TSP instance is straightforward, by letting $d_{ij}$ denote the cost of processing a job of type $i$ together with the cost of switching from type $i$ to type $j$. (Note that $d_{ii}$ is not necessarily smaller than $d_{ij}$ for $i \neq j$; in some applications it is beneficial to switch between job types.)
Emmons and Mathur~\cite{EmmonsMathur1995} also describe an application of MV-TSP to the no-wait flow shop problem.

A different kind of application comes from geometric approximation. 
A standard technique to approximate geometric optimization problems is to reduce the size of the input by grouping certain input points together. Each group is then replaced by a single representative, and the reduced instance is solved exactly.
For example, we may snap input points to nearby grid points, if doing so does not significantly affect the objective cost.
Recently, this technique was used by Kozma and M\"omke, to give an efficient polynomial-time approximation scheme (EPTAS) for the {\sc Maximum Scatter TSP} in doubling metrics~\cite{KozmaMomke2017}, addressing an open question of Arkin et al.~\cite{ArkinEtAl1999}.
In this case, the reduced problem is exactly the MV-TSP.
Yet another application of MV-TSP is in settling the parameterized complexity of finding a Hamiltonian cycle in a graph class with restricted neighborhood structure~\cite{Lampis2012}.

To the best of our knowledge, MV-TSP was first considered in 1966 by Rothkopf~\cite{Rothkopf1966}.
In 1980, Psaraftis~\cite{Psaraftis1980} gave a dynamic programming algorithm with run time $O(n^2 \cdot \prod_{i \in V}{(k_{i}+1)})$.
Observe that this quantity may be as high as $(k/n+1)^n$, which is prohibitive even for moderately large values of $k$.
In 1984, Cosmadakis and Papadimitriou~\cite{CosmadakisPapadimitriou1984} observed that MV-TSP can be decomposed into a connectivity subproblem and an assignment subproblem.
Taking advantage of this decomposition, they designed a family of algorithms, the best of which has run time $O^{\ast}(n^{2n} 2^n + \log{k})$.\footnote{Here, and in the following, the $O^{\ast}(\cdot)$ notation is used to suppress a factor bounded by a polynomial in $n$.}
This result can be seen as an early example of \emph{fixed-parameter tractability}, where the rapid growth in complexity is restricted to a certain parameter. 

The algorithm of Cosmadakis and Papadimitriou is, to date, the fastest solution to MV-TSP.\footnote{It may seem that a linear dependence on the length $k$ of the tour is necessary even to output the result.
    Observe however, that a tour can be compactly represented by collapsing cycles and storing them together with their multiplicities.}
Its analysis is highly non-trivial, combining graph-theoretic insights and involved estimates of various combinatorial quantities. The analysis of the Cosmadakis-Papadimitriou algorithm is not known to be tight, 
but a lower bound of the form $2^{\Omega(n \log{n})}$ is known to hold for its run time~\cite[p.\ 104]{CosmadakisPapadimitriou1984}. 
Similarly, in the space requirement of the algorithm, a term of the form $2^{\Omega(n \log{n})}$ appears hard to avoid. 
While it extends the tractability of TSP to a new range of parameters, the usefulness of the Cosmadakis-Papadimitriou algorithm is limited by its superexponential\footnote{In this paper the term \emph{superexponential} always refers to a quantity of the form $2^{\omega{(n)}}$.} dependence on $n$ in the run time, with 
the issue of superexponential \emph{space} perhaps even more worrisome.  

There have been further studies of the MV-TSP problem.
Van der Veen and Zhang \cite{vanderVeenZhang1996} discuss a problem equivalent to MV-TSP, called \emph{K-group TSP}, and describe an algorithm with polylogarithmic dependence on the number $k$ of visits, similarly to Cosmadakis and Papadimitriou.
The value~$n$ however is assumed constant, and its effect on the run time is not explicitly computed (the dependence can be seen to be superexponential). 
Finally, Grigoriev and van de Klundert~\cite{GrigorievvandeKlundert2006} give an ILP formulation for MV-TSP with $O(n^2)$ variables.
Applying Kannan's improvement~\cite{Kannan1983} of Lenstra's algorithm~\cite{Lenstra1983} for solving fixed-dimensional ILPs to this formulation yields an algorithm with run time $n^{O(n^2)} \log{k}$.
Further ILP formulations for MV-TSP are due to Sarin et al.~\cite{SarinEtAl2011} and Aguayo et al.~\cite{AguayoEtAl2018}, both of which again require superexponential time to be solved by standard algorithms. For further details about the history of the MV-TSP problem we refer to the TSP textbook of Gutin and Punnen~\cite[\S\,11.10]{GutinEtAl2002}.

\paragraph*{Our results.}
Our main result improves both the time and space complexity of the best known algorithm for MV-TSP, the first improvement in over 35 years.
Specifically, we show that a \emph{logarithmic} dependence on the number $k$ of visits, a \emph{single-exponential} dependence on the number~$n$ of cities, and a \emph{polynomial} space complexity are simultaneously achievable.
Moreover, while we build upon ideas from the previous best approach, our algorithm is arguably easier to describe, easier to implement, and easier to analyse than its predecessor.
To introduce the techniques step by step, we describe \emph{three} algorithms for solving MV-TSP.
We refer to the three algorithms, after their core subroutines, as \textsc{enum-MV}, \textsc{dp-MV}, and \textsc{dc-MV}. The acronyms stand for enumeration, dynamic programming, and divide and conquer.
We also describe an improved variant of the third algorithm, called \textsc{dc-MV2}. All our algorithms are deterministic.
Their complexities are summarized in Theorem~\ref{mainthm}, proved in~\S\,\ref{sec:improvedalgorithmsforthemanyvisitstsp}. 

\begin{theorem}
    \label{mainthm}  \ \\
\vspace{-0.2in}
\begin{enumerate}[(i)]
\item \textsc{enum-MV} solves MV-TSP using space $O(n^2)$, in time $O^{\ast}(n^n + \log k)$. 

\item \textsc{dp-MV} solves MV-TSP using space $O(5^n)$, in time $O^{\ast}(5^n + \log k)$. 

\item \textsc{dc-MV2} solves MV-TSP using space $O(n^2)$, in time $O^{\ast}(16^{n + o(n)} + \log k)$. 
\end{enumerate}

\end{theorem}
The Exponential-Time Hypothesis (ETH)~\cite{ImpagliazzoEtAl2001, Impagliazzo2} implies that TSP cannot be solved in $2^{o(n)}$, i.e.\ sub-exponential, time.
Under this hypothesis, the run time of our algorithm \textsc{dc-MV2} is asymptotically optimal for MV-TSP, up to the base of the exponential.
Further, note that the space requirement of {\sc dc-MV2} is also (essentially) optimal, as a compact solution encodes for each of the $\Omega(n^2)$ edges the number $t$ of times that this edge is traversed by an optimal tour. Note that throughout the paper we assume that each multiplicity can be stored in a constant number of machine words; if this is not the case, e.g.\ if $k$ is exponential in $n$, a factor $O(\log{k})$ should be applied to the given space bounds.

Our result leads to improvements in applications where MV-TSP is solved as a subroutine. For instance, as a corollary of Theorem~\ref{mainthm}\emph{(iii)}, the approximation scheme for Maximum Scatter TSP~\cite{KozmaMomke2017} can now be implemented in space \emph{polynomial} in the error parameter $\varepsilon$. 

It is interesting to contrast our results for MV-TSP with recent results for the \emph{$r$-simple path} problem, where a path of length $k$ is sought that visits each vertex \emph{at most} $r$ times.
For that problem, the fastest known algorithms---due to Abasi et al.~\cite{AbasiEtAl2014} and Gabizon et al.~\cite{GabizonEtAl2015}---have run time~\emph{exponential} in $k/r \cdot \log{r}$, and such exponential dependence is  necessary assuming ETH.

\paragraph*{Overview of techniques.}
The Cosmadakis-Papadimitriou algorithm is based on the following high-level insight, common to most work on the TSP problem, whether exact or approximate.
The task of finding a valid tour may be split into two separate tasks: (1) finding a structure that connects all vertices, and (2) augmenting the structure found in (1) in order to ensure that each vertex is visited the required number of times. 
Indeed, such an approach is also used in the well-known $3/2$-approximation algorithm of Christofides for metric TSP~\cite{Christofides1976}.
There, the structure that guarantees connectivity is a minimum spanning tree, and ``visitability'' is  enforced by the addition of a perfect matching that connects odd-degree vertices, ensuring that all vertices have even degree, and can thus be entered and exited, as required.

In the case of MV-TSP, Cosmadakis and Papadimitriou ensure connectivity (part (1)) by finding a \emph{minimal connected Eulerian digraph} on the set $V$ of input vertices.
Indeed, such a digraph must be part of every solution, since a tour must balance every vertex (equal out-degree and in-degree), and all vertices must be mutually reachable.
Minimality is meant here in the sense that no proper subgraph is Eulerian, and is required only to reduce the search space. 
Assuming that a connected Eulerian subdigraph of the solution is found, it needs to be extended to the edge set of a valid tour (part (2)).
If this is done with the cheapest possible set of edges, then the optimum must have been found.
This second step amounts to solving a transportation problem, which takes polynomial time.

The first step, however, requires us to implicitly consider all possible minimal Eulerian digraphs.
As it is $\mathsf{NP}$-complete to test the non-minimality of an Eulerian digraph~\cite{PapadimitriouYannakakis1981}, the authors relax minimality and suggest the use of heuristics for pruning out non-minimal instances in practice.
On the other hand, they obtain a saving in run time by observing that among all digraphs with the same \emph{degree sequence} only one with smallest cost needs to be considered.
Otherwise, in the final tour, the connected Eulerian subdigraph could be swapped with a cheaper one, while maintaining the validity of the tour.

Cosmadakis and Papadimitriou thus iterate over feasible degree sequences of connected Eulerian digraphs; for each such degree sequence they construct the cheapest realization (which may not be minimal) by dynamic programming; finally they construct, for each obtained Eulerian digraph, the cheapest extension to a valid tour, by solving a transportation problem.
The cheapest tour found over all iterations is returned as the solution. 
Iterating and optimizing over these structures is no easy task, and Cosmadakis and Papadimitriou invoke a number of graph-theoretic and combinatorial insights. For estimating the total cost of their procedure a sophisticated global counting argument is developed.

The key insight of our approach is that the machinery involving Eulerian digraphs is \emph{not necessary} for solving MV-TSP.
To ensure connectivity (i.e.\ task (1) above), a \emph{directed spanning tree} is sufficient. 
This may seem surprising, as a directed tree fails to satisfy the main property of connected Eulerian digraphs, \emph{strong connectivity}. 
Observe however, that a collection of directed edges with the same out-degrees and in-degrees as a valid MV-TSP tour is itself a valid tour, as long as the digraph determined by the collection of edges is weakly connected. 
Requiring the solution to contain a tree is sufficient to avoid the case of disjoint cycles.
The fact that the tree can be assumed to be \emph{rooted}, i.e.\ all of its edges are directed away from some vertex, follows from the strong connectedness of the tour.\footnote{We thank Andreas Bj{\"o}rklund for the latter observation which led to an improved run time and a simpler correctness argument.} 

Directed spanning trees are easier to enumerate and optimize over than Eulerian digraphs; this fact alone explains the reduced complexity of our approach.
However, to obtain our main result, further ideas are needed.
In particular, we find the cheapest directed spanning tree that is feasible for a given degree sequence, first by dynamic programming, then by a recursive partitioning of trees, based on centroid-decompositions. 

\section{Improved algorithms for the Many-Visits TSP}
\label{sec:improvedalgorithmsforthemanyvisitstsp}
In this section we describe and analyse our three algorithms.
The first, \textsc{enum-MV} is based on exact enumeration of trees (\S\,\ref{secenum}), the second, \textsc{dp-MV} uses a dynamic programming approach to find an optimal tree (\S\,\ref{secdp}), and the third, \textsc{dc-MV} is based on divide and conquer (\S\,\ref{secdc} and \S\,\ref{secdc2}). 
Before presenting the algorithms, we introduce some notation and structural observations that are subsequently used (\S\,\ref{sec:treestoursanddegreesequences}).

\subsection{Trees, tours, and degree sequences}
\label{sec:treestoursanddegreesequences}
Let $V$ be a set of vertices.
We view a \emph{directed multigraph} $G$ with vertex set $V$ as a \emph{multiset} of edges (i.e.\ elements of $V \times V$).
Accordingly, self-loops and multiple copies of the same edge are allowed.
The \emph{multiplicity} of an edge $(i,j)$ in a directed multigraph $G$ is denoted $m_G(i,j)$.
The \emph{out-degree} of a vertex $i \in V$ is $\dout_G(i) = \sum_{j \in V}{m_G(i,j)}$, the \emph{in-degree} of a vertex $i \in V$ is $\din_G(i) = \sum_{j \in V}{m_G(j,i)}$.
Given edge costs $d: V \times V \rightarrow \mathbb{N} \cup \{\infty\}$, the \emph{cost} of $G$ is simply the sum of its edge costs, i.e.\ $\cost(G) = \sum_{i,j \in V} m_G(i,j) \cdot d(i,j)$. 

For two directed multigraphs $G$ and $H$ over the same vertex set $V$, let~$G+H$ denote the directed multigraph obtained by adding the corresponding edge multiplicities of $G$ and~$H$.
Observe that as an effect, out-degrees and in-degrees are also added pointwise.
Formally, $m_{G+H}(i,j) = m_G(i,j) + m_H(i,j)$, and $\dout_{G+H}(i) = \dout_G(i) + \dout_H(i)$, and $\din_{G+H}(i) = \din_G(i) + \din_H(i)$, for all $i,j \in V$. The relation $\cost(G+H) = \cost(G) + \cost(H)$ clearly holds.

In the following, for a directed multigraph $G$ we refer to its \emph{underlying graph}, i.e.\ to the undirected graph consisting of those edges $\{i,j\}$ for which $m_G(i,j) + m_G(j,i) \geq 1$. 

Consider a tour $C = (x_1, \dots, x_k) \in V^k$.
We refer to the unique directed multigraph $G$ consisting of the edges $(x_1, x_2), \dots, (x_{k-1},x_k), (x_k,x_1)$ as the \emph{edge set of} $C$.
We state a simple but crucial observation.
\begin{lemma}
\label{lem1}
  Let $G$ be a directed multigraph over $V$ with out-degrees $\dout_G(\cdot)$ and in-degrees~$\din_G(\cdot)$. 
  Then $G$ is the edge set of a tour that visits each vertex $i \in V$ exactly $k_i$ times
  if and only if both of the following conditions hold:
  \begin{enumerate}[(i)]
    \item the underlying graph of $G$ is connected, and
    \item for all $i \in V$, we have $\dout_G(i) = \din_G(i) = k_i$.
  \end{enumerate}
\end{lemma}
\begin{proof}
  The fact that connectedness of $G$ and $\dout_G(i) = \din_G(i)$ is equivalent with the existence of a tour that uses each edge of $G$ exactly once is the well-known ``Euler's theorem''. (See~\cite[Thm~1.6.3]{BangJensenGutin2002} for a short proof.)
  Clearly, visiting each vertex $i$ exactly $k_i$ times is equivalent with the condition that the tour contains $k_i$ edges of the form $(\cdot,i)$ and $k_i$ edges of the form $(i,\cdot)$.
\end{proof}

Moreover, given the edge set $G$ of a tour $C$, a tour $C'$ with edge set $G$ can easily be recovered.
This amounts to finding an Eulerian tour of $G$, which can be done in time linear in the length $k$ of the tour.
To avoid a linear dependence on $k$, we can apply the algorithm of Grigoriev and van de Klundert~\cite{GrigorievvandeKlundert2006} that constructs a compact representation of $C'$ in time $O(n^4 \log{k})$.
As the edge sets of $C$ and $C'$ are equal, $C'$ also visits each $i \in V$ exactly $k_i$ times and $\cost(C') = \cost(C) = \cost(G)$.
Thus, in solving MV-TSP we only focus on finding a minimum cost directed multigraph whose underlying undirected graph is connected, and whose degrees match the multiplicities required by the problem. 

A \emph{directed spanning tree} of $V$ is a tree with vertex set $V$ whose edges are directed \emph{away} from some vertex $r \in V$; in other words, the tree contains a directed path from $r$ to every other vertex in $V$. 
(Directed spanning trees are alternatively called \emph{branchings}, \emph{arborescences}, or \emph{out-trees}.)
We refer to the vertex $r$ as the \emph{root} of the tree.
We observe that every valid tour contains a directed spanning tree. 

\begin{lemma}
\label{lem3}
  Let $G$ be the edge set of a tour of~\;$V$ with arbitrary non-zero multiplicities, and let $r \in V$ be an arbitrary vertex.
  Then there is a directed spanning tree $T$ of~\;$G$ rooted at $r$, and a directed multigraph~$X$, such that $G = T + X$. 
\end{lemma}
\begin{proof}
  We choose $T$ to be the single-source \emph{shortest path tree} in $G$ with source $r$.
  More precisely, let $d(v)$ denote the \emph{distance} (i.e.\ number of edges) from $r$ to $v$ in $G$. Observe that in a valid tour all vertices are mutually reachable, so $d(v)$ is finite for all $v \in V$. Let us now build $T$, by adding, for each vertex $v \in V \setminus \{r\}$, an edge $(w,v)$, where $w \in V$ is an arbitrary vertex such that $d(w) = d(v) - 1$. Such a $w$ must exist, as the predecessor of $v$ on a shortest path from $r$ has this property. The fact that $d(\cdot)$ strictly increases along edges ensures that $T$ is cycle-free, i.e.\ a tree. 
\end{proof}

We can thus split the MV-TSP problem into finding a directed spanning tree $T$ with an arbitrary root $r$ and an extension $X$, such that $T+X$ is a valid tour.
We claim that in the decomposition $G = T + X$ of an optimal tour $G$, both $T$ and~$X$ are optimal with respect to their degree sequences.
\begin{lemma}
\label{lem2}
  Let $G$ be the edge set of an optimal tour for MV-TSP, let $T$ be a directed spanning tree, and let $X$ be a directed multigraph such that $G=T+X$.
  Then, $T$ has the smallest cost among all directed spanning trees with degrees $\dout_T(\cdot)$ and $\din_T(\cdot)$, and $X$ has the smallest cost among all directed multigraphs with degrees $\dout_X(\cdot)$ and $\din_X(\cdot)$.
\end{lemma} 
\begin{proof}
  Suppose there is a directed spanning tree $T'$ such that $\cost(T') < \cost(T)$, and $\dout_{T'}(i) = \dout_T(i)$, and $\din_{T'}(i) = \din_T(i)$ for all $i \in V$.
  But then $T' + X$ is connected, has the same degree sequence as $G$, while $\cost(T' + X) < \cost(G)$, contradicting the optimality of $G$.

  Similarly, suppose there is a directed multigraph $X'$ such that $\cost(X') < \cost(X)$, and $\dout_{X'}(i) = \dout_X(i)$, and $\din_{X'}(i) = \din_X(i)$ for all $i \in V$.
  But then $T + X'$ is connected, has the same degree sequence as $G$, while $\cost(T + X') < \cost(G)$, contradicting the optimality of $G$.
\end{proof}

Next, we characterize the  feasible  degree sequences of directed spanning trees.\footnote{The feasibility of degree sequences for various graph classes is a well-studied subject, see e.g.\ \cite{Havel1955, Hakimi1962, ErdosGallai1960, Fulkerson1960, KleitmanWang1973, Kim2009, Berger2014}. The simple condition we state is similar to the condition for \emph{undirected} trees, given by Berge~\cite[page 117]{Berge1973}.}

\begin{lemma}
\label{lem:deg_seq_dir}
  Let $V = \{x_1, \dots, x_n\}$ be a set of vertices, where $n \geq 2$.
  There is a directed spanning tree of~\;$V$ with root $x_1$ whose out-degrees and in-degrees are respectively $\dout(\cdot)$ and $\din(\cdot)$, if and only if
  \begin{enumerate}[(i)]
    \item $\din(x_1) = 0$,
    \item $\din(x_i) = 1$ for $i=2,\dots,n$,
    \item $\dout(x_1) > 0$, and
    \item $\sum_i \dout(x_i) = n-1$.
  \end{enumerate}
\end{lemma}
\begin{proof}
  In the forward direction, in a directed spanning tree all non-root vertices have exactly one parent, proving (i) and (ii).
  The root must have at least one child (iii), and the total number of edges is $n-1$, proving (iv).

  In the backward direction, we argue by induction on $n$.
  In the case $n=2$, we have $\din(x_1) = \dout(x_2) = 0$, and $\din(x_2) = \dout(x_1) = 1$, hence an edge $(x_1,x_2)$ satisfies the degree requirements.
  
  Consider now the case of $n>2$ vertices.
  From $(ii)$--$(iv)$ it follows that for some $k\in\{2,\dots,n\}$, we have $\din(x_k) = 1$ and $\dout(x_k) = 0$, i.e.\ $x_k$ is a leaf.
  
  Let $x_j$ be a vertex for some $j\in\{1,\dots,n\}\setminus\{k\}$ such that $\dout(x_j) \geq 1$, and $\din(x_j) + \dout(x_j) \geq 2$; by $(ii)$--$(iv)$ there must be such a vertex.
  We decrease $\dout(x_j)$ by one.
  Conditions $(i)$--$(iv)$ clearly hold for $V \setminus \{x_k\}$.
  By induction, we can build a tree on $V \setminus \{x_k\}$, and attach $x_k$ to this tree as a leaf, with the edge~$(x_j,x_k)$.
\end{proof}

Let $\DS(n)$ denote the number of different pairs of sequences $(\delta'_1, \dots, \delta'_n), (\delta''_1, \dots, \delta''_n)$, that are feasible for a directed spanning tree, i.e.\ for vertex set $V = \{x_1, \dots, x_n\}$, for some directed spanning tree $T$ with root $x_1$, we have $\dout_T(x_i) = \delta'_i$, and $\din_T(x_i) = \delta''_i$, for all $i \in \{1,\hdots,n\}$. 
By Lemma~\ref{lem:deg_seq_dir}, $\DS(n)$ equals the number of ways to distribute $n-1$ out-degrees to $n$ vertices, such that a designated vertex has non-zero out-degree.
This task is the same as distributing $n-2$ balls arbitrarily into $n$ bins, of which there are ${2n-3 \choose n-1} = O(4^n)$ ways.

\subsection{{{{\sc enum-MV}}}: polynomial space and superexponential time}
\label{secenum}
Given the vertex set $V = \{1,\dots,n\}$, multiplicities $k_i \in\mathbb N$ and cost function $d$, we wish to find a tour of minimum cost that visits each $i \in V$ exactly $k_i$ times. 
From Lemma~\ref{lem3}, our first algorithm presents itself.
It simply iterates over all directed spanning trees $T$ with vertex set $V$, and extends each of them optimally to a valid tour $C_T$. 
Among all valid tours constructed, one with smallest cost is returned (Algorithm~\ref{alg:enumMV1}). 

This simple algorithm already improves on the previous best run time (although it is still superexponential), and reduces the space requirement from superexponential to polynomial.

\begin{algorithm}
  \caption{{\sc enum-MV} for solving MV-TSP using enumeration}\label{alg:enumMV1}
  \begin{algorithmic}[1]
    \Statex \textbf{Input:} Vertex set $V$, cost function $d$, multiplicities $k_i$.
    \Statex \textbf{Output:} A tour of minimum cost that visits each $i \in V$ exactly $k_i$ times. 
    \For {each directed spanning tree $T$ with root $r \in V$}        
      \State Find minimum cost directed multigraph $X$ such that for all $i \in V$:
	\Statex \quad \quad $\dout_{X}(i) = k_i - \dout_T(i)$,
	\Statex \quad \quad $\din_{X}(i) = k_i - \din_T(i)$.
	\Statex \ \quad Let $C_T \leftarrow T+X$.
        \EndFor
        \State \textbf{return} $C_T$ with smallest cost.
    \end{algorithmic}
\end{algorithm}

The correctness of the algorithm is immediate: from Lemma~\ref{lem3} it follows that all $C_T$'s considered are valid (connected, and with degrees matching the required multiplicities), and by Lemma~\ref{lem2}, the optimal tour $C$ must be considered during the execution. 

The iteration of Line 1 requires us to enumerate all labeled trees with vertex set $V$.
There are~$n^{n-2}$ such trees~\cite{Cayley1889}, and standard techniques can be used to enumerate them with a constant overhead per item (see e.g.\ Kapoor and Ramesh~\cite{KapoorRamesh1995}).
For each considered tree we orient the edges in a unique way, away from $r$. 

Let $T$ be the current tree.
In Line 2 we find a minimum cost directed multigraph~$X$, with given out-degree and in-degree sequence, such as to extend $T$ into a valid tour.
If, for some vertex $i$, it holds that $\din_T(i) > k_i$ or $\dout_T(i) > k_i$, we proceed to the next spanning tree, since the current tree cannot be extended to a valid tour. Observe that this may happen only if $k_i < n-1$, for some $i \in V$. 
Otherwise, we find the optimal $X$ by solving a transportation problem in polynomial time.
Throughout the execution we remember the best tour found so far, which we output in the end.
We describe next the transportation subroutine, which is common to all our algorithms, and is essentially the same as in the Cosmadakis-Papadimitriou algorithm.

\paragraph*{The transportation problem.}
The subproblem we need to solve is finding a minimum cost directed multigraph $X$ over vertex set $V$, with given out-degree and in-degree requirements.
We can map this problem to an instance of the Hitchcock transportation problem~\cite{Hitchcock1941}, where the goal is to transport a given amount of goods from $N$ warehouses to $M$ outlets with given pairwise shipping costs.
Note that this is a special case of the more general min-cost max-flow problem~\cite{EdmondsKarp1970}.

More precisely, let us define a digraph with vertices $\{s,t\} \cup \{s_i,t_i ~|~ i \in V \}$.
Edges are\linebreak $\{(s,s_i),(t_i,t) ~|~ i \in V\}$, and $\{(s_i,t_j) ~|~ i,j \in V\}$.
We set cost $0$ to edges $(s,s_i)$ and $(t_i,t)$ and cost~$d_{ij}$ (i.e.\ the costs given in the MV-TSP instance) to $(s_i,t_j)$. We set capacity $\infty$ to edges~$(s_i,t_j)$ and capacities $k_i - \dout_T(i)$ to $(s,s_i)$, and $k_i - \din_T(i)$ to $(t_i,t)$.
The construction is identical to the one used by Cosmadakis and Papadimitriou, apart from the fact that in our case the capacity of $(s,s_i)$ may be different from the capacity of $(t_i,t)$.
Observe that the sum of capacities of $(s,s_i)$-edges equals the sum of capacities of $(t_i,t)$-edges over all $i \in V$.
Thus, a maximal $s-t$ flow saturates all these edges. 
The amount of flow transmitted on the edge $(s_i,t_j)$ gives the multiplicity of edge $(i,j)$ in the sought after multigraph, for all $i,j \in V$.
A minimum cost maximum flow clearly maps to a minimum cost edge set with the given degree constraints.

In the Cosmadakis-Papadimitriou algorithm, the transportation subproblems are solved via the scaling method of Edmonds and Karp~\cite{EdmondsKarp1970}. This algorithm proceeds by solving $O(\log{k})$ approximate versions of the problem, where the costs are the same as in the original problem, but the capacities are scaled (i.e.\ divided and rounded down) by a factor $2^p$ for $p = \ceil{\log k}, \dots, 0$.
Each approximate problem is solved in $O(n^3)$ time, by performing flow augmentations on the optimal flow found in the previous approximation, multiplied by two. The overall run time for solving the described transportation problem is therefore $O(n^3 \log{k})$.

Cosmadakis and Papadimitriou describe an improvement which also applies for our case.
Namely, they show that the total run time for solving several instances with the same costs can be reduced, if the capacities on corresponding edges in two different instances may differ by at most $n$. 
The strategy is to solve all but the last $\log{n}$ approximate problems only once, as these are (essentially) the same for all instances.
For different instances we only need to solve the last $\log{n}$ approximate problems, i.e.\ at the finest levels of approximation.
This gives a run time of $O(n^3 \log{k})$ for solving the ``master problem'', and $O(n^3 \log{n})$ for solving each individual instance.
We refer to Cosmadakis and Papadimitriou~\cite{CosmadakisPapadimitriou1984} as well as to Edmonds and Karp~\cite{EdmondsKarp1970} for details.

In our case, the different instances of the transportation problem are for finding the directed multigraphs $X$ for different trees $T$.
Each of these instances agree in the underlying graph and cost function, and may differ only in the capacities.
As the maximum degree of each tree $T$ is at most~$n-1$, the differences are bounded, as required.

As an alternative to the Edmonds-Karp algorithm, we may solve the arising transportation problems with a \emph{strongly polynomial} algorithm, e.g.\ the one by Orlin~\cite{Orlin} or its extension due to Kleinschmidt and Schannath~\cite{Kleinschmidt}. Note that these algorithms were not yet available when Cosmadakis and Papadimitriou obtained their result. The run time for the transportation subproblem then becomes $O(n^3 \log{n})$, i.e.\ independent of $k$.   
Such an improvement is likely of theoretical interest only; furthermore, it assumes that operations on the multiplicities $k_i$ take constant time. If this assumption is unrealistic, e.g.\ if $k$ is exponential in $n$, we may fall back to the Edmonds-Karp algorithm, with the term $O(n^3 \log{k})$ added to the run time. 
 
\paragraph*{Analysis of Algorithm~\ref{alg:enumMV1}.}
We iterate over all $O(n^{n-2})$ directed spanning trees and solve a transportation problem for each, with run time $O(n^3 \log{n})$.
The total run time $O(n^{n+1} \log{n})$ follows. 
The space requirement of the algorithm is dominated by that of solving a (single) transportation problem, and of storing the edge set of a single tour (apart from minor bookkeeping).

\subsubsection{Improved enumeration algorithm}
We can slightly improve the run time of {\sc enum-MV} by observing that the solution of the transportation problem depends only on the degree sequence of the current tree $T$, and not the actual edges of $T$.
Therefore, different trees with the same degree sequence can be extended in the same way.
Observe that several trees may have the same degree sequence; 
in an extreme case, all $(n-2)!$ simple Hamiltonian paths with the same endpoints have the same degree sequence.

\begin{algorithm}
  \caption{{\sc enum-MV} for solving MV-TSP using enumeration (improved)}\label{alg:enumMV2}
  \begin{algorithmic}[1]
    \Statex \textbf{Input:} Vertex set $V$, cost function $d$, multiplicities $k_i$.
    \Statex \textbf{Output:} A tour of minimum cost that visits each $i \in V$ exactly $k_i$ times. 
    \For {each feasible degree sequence $\dout(\cdot), \din(\cdot)$ of a directed spanning tree with root $r$}        

        \State Find minimum cost directed multigraph $X$ such that for all $i \in V$:
        \Statex \ \ \quad \quad \quad $\dout_{X}(i) = k_i - \dout(i)$,
	    \Statex \ \ \quad \quad \quad $\din_{X}(i) = k_i - \din(i)$.
        \For {each directed spanning tree $T$ with $\dout_T = \dout$ and $\din_T = \din$} 
	    \Statex \ \ \ \quad \quad Let $C_T \leftarrow T+X$.
      \EndFor
   \EndFor
   \State \textbf{return} $C_T$ with smallest cost.
  \end{algorithmic}
\end{algorithm}

Algorithm~\ref{alg:enumMV2} implements this idea, iterating over all trees, grouped by their degree sequences.
It solves the transportation problem only once for each degree sequence (there are $\DS(n)$ of them). 

The correctness is immediate, as all directed spanning trees $T$ are still considered as before.
Assume that we can iterate over all feasible degree sequences, and all corresponding trees with $O(n)$ overhead per item. By Lemma~\ref{lem:deg_seq_dir}, the first task only requires us to consider all ways of distributing $n-2$ out-degrees among $n$ vertices. For completeness, we describe a procedure for this task in Appendix~\ref{combalg}. We give the subroutine for the second task in \S\,\ref{secsub2}.
We thus obtain the run time $O(n^{n-1}  + \DS(n) \cdot n^4 \log{n}) = O(n^{n-1})$.
The space requirement is asymptotically unchanged.

\subsubsection{Generating trees by degree sequence}
\label{secsub2}
In the proof of Lemma~\ref{lem:deg_seq_dir}, we generate \emph{one} directed tree from its degree sequence.
In this subsection we show how to generate \emph{all} trees for a given degree sequence (Algorithm~\ref{alg:degree_sequence_tree}). 

The initial call to the \textsc{buildTree}() procedure is with a feasible input degree sequence $(\dout, \din)$ and an empty ``stub'' ($\mathbf{x}^\text{stub} = \{\}$) as arguments.
The algorithm finds the first \emph{unattached} vertex that is either a leaf of the final tree or whose subtree is already complete (that is, $\dout(x_i) = 0$), then it finds all possibilities for attaching $x_i$ to the rest of the tree. Observe that the fact that $x_i$ has no more capacity for outgoing edges prevents cycles. 
The procedure is then called recursively with modified arguments: edge $(x_j,x_i)$ is added to the stub and the out-degree of $x_j$ and the in-degree of~$x_i$ are decremented. 

\begin{algorithm}
  \caption{{\sc buildTree} for generating all directed trees for a given degree sequence\label{alg:degree_sequence_tree}}
  \begin{algorithmic}[1]
    \Statex \textbf{Input:} Degree sequence $(\dout, \din)$ and forest $\mathbf{x}^\text{stub}$.
    \Statex \textbf{Output:} A \emph{generator} of all trees with degree sequence $(\dout, \din)$, attached to $\mathbf{x}^\text{stub}$.
    \Procedure{{\sc buildTree}}{$\dout$, $\din$, $\mathbf{x}^\text{stub}$}
    \If{$\sum \din(\cdot) = 1$ }
      \State Let $i$ be the index where $\din(x_i)=1$.
      \State \textbf{yield} $x^\text{stub} \cup \{(x_1,x_i)\}$
    \Else 
      \State Let $i$ be the first index where $\dout(x_i) = 0$ and $\din(x_i) = 1$. 
      \For{every index $j \neq i$, such that $\dout(x_j) + \din(x_j)\geq 2$} \hfill$\triangleright$ attach $x_i$ to parent $x_j$
        \State $\Bigl( {\dout}', {\din}' \Bigr)  \leftarrow  \Bigl( \dout, \din \Bigr) $
        \State ${\dout}' (x_j) \leftarrow {\dout}' (x_j) - 1$ 
        \State  ${\din}' (x_i) \leftarrow 0$ 
        \State \textsc{buildTree}(${\dout}' $, ${\din}' $, $\mathbf{x}^\text{stub} \cup \{(x_j,x_i)\}$)
      \EndFor
    \EndIf
    \EndProcedure
  \end{algorithmic}
\end{algorithm}

At each recursive level there are as many new calls as possible candidates for the next edge, with both degree demands decreased by one, and with the stub gaining one additional edge. During the intermediate calls the stub may be disconnected, i.e.\ it represents a forest.
At the $(n-1)$-th level exactly two degree-one vertices remain, say, $x_i$ with in-degree $1$ and the root~$x_1$ with out-degree $1$.
Adding the edge $(x_1,x_i)$ finishes the construction. 

Observe that there are no ``dead ends'' during this process, i.e.\ every call of \textsc{buildTree} eventually results in a valid directed tree in the last level of the recursion, and there are no discarded graphs during the process. The argument mirrors the one in the proof of Lemma~\ref{lem:deg_seq_dir}. In every call to \textsc{buildTree}, the degree sequence $(\delta^\text{out}, \delta^\text{in})$ encodes a valid directed tree on the subset of vertices with non-zero in-degree, together with the root $x_1$. Let $h$ denote the number of vertices with non-zero in-degree. We argue by induction on $h$ that the construction succeeds. For $h=1$ (Lines 2--4) this is clear. For $h\geq 2$ (Lines 6--11), we attach $x_i$ to a vertex $x_j$ of a valid tree, constructed by the recursive call with $h-1$ vertices with non-zero in-degree. 

For every directed tree with the required degree-sequence, repeatedly peeling off the lowest-index leaf gives a valid edge-insertion order that is found by the algorithm. (The order in which vertices are considered is essentially the \emph{Pr\"ufer sequence}~\cite{moon1970counting} of the tree.) Finally, as the recursive calls correspond to attaching $x_i$ to \emph{different} parents, no tree is constructed more than once. 

Note that Algorithm~\ref{alg:degree_sequence_tree} does not return the list of trees realising a degree sequence $(\dout, \din)$ at once, it is instead a 
\emph{generator} that serves as the head of the loop in Step~3 of Algorithm~\ref{alg:enumMV2}, outputting a different tree in every iteration. 
We used the keyword \textbf{yield} instead of \textbf{return} to indicate that the algorithm does not terminate after ``returning'' the first result.
   
\subsection{{\sc dp-MV}: exponential space and single-exponential time}
\label{secdp}
Next, we improve the run time to single-exponential, by making use of Lemma~\ref{lem2}.
Specifically, we observe that for every feasible degree sequence only the tree with minimum cost needs to be considered. 
The enumeration in Algorithm~\ref{alg:degree_sequence_tree} can easily be modified to return, instead of all trees, just one with smallest cost, this, however, would not improve the asymptotic run time.
Instead, in this section we describe a dynamic programming algorithm resembling the algorithms by Bellman~\cite{Bellman1962}, and Held and Karp~\cite{HeldKarp1962}, for directly computing the best directed tree for a given degree sequence. 

The outline of the algorithm is shown in Algorithm~\ref{alg:outline}, and it is identical for {\sc dp-MV} and {\sc dc-MV}, described in \S\,\ref{secdc} and \S\,\ref{secdc2}.
The algorithms {\sc dp-MV} and {\sc dc-MV} differ in the way they find the minimum cost directed tree, i.e.\ Line~2 of the generic Algorithm~\ref{alg:outline}.

\begin{algorithm}
  \caption{Solving MV-TSP by optimizing over trees (common for {\sc dp-MV}, {\sc dc-MV} and {\sc dc-MV2}).\label{alg:outline}}
  \begin{algorithmic}[1]
    \Statex \textbf{Input:} Vertex set $V$, cost function $d$, multiplicities $k_i$.
    \Statex \textbf{Output:} A tour of minimum cost that visits each $i \in V$ exactly $k_i$ times. 
    \For{each feasible degree sequence $\dout(\cdot), \din(\cdot)$ of a directed spanning tree}        
      \State Find a minimum-cost directed tree $T$ with $\din_T = \din$ and $\dout_T = \dout$.
      \State Find a minimum-cost directed multigraph $X$ such that for all $i \in V$:
      \Statex \ \ \quad \quad \quad $\din_{X}(i) = k_i - \din_T(i)$,
      \Statex \ \ \quad \quad \quad $\dout_{X}(i) = k_i - \dout_T(i)$.
      \Statex \ \ \ \quad \quad Let $C_T \leftarrow T+X$.
    \EndFor
    \State \textbf{return} $C_T$ with smallest cost.
  \end{algorithmic}
\end{algorithm}

The dynamic programming approach ({\sc dp-MV}) resembles Algorithm~\ref{alg:degree_sequence_tree} for iterating over all directed trees with a given degree sequence.
Specifically, let $(\dout, \din)$ be the degree sequence for which we wish to find a minimum-cost tree.
We build a dynamic programming table $\T$ that holds an optimal tree and its cost for \emph{every} feasible degree sequence.
The solution can thus be read from $\T[\dout,\din]$.
Observe that specifying a degree sequence allows us to restrict the problem to arbitrary subsets of~$V$, by simply setting the degrees of non-participating vertices to zero.

To compute $\T[\dout,\din]$, we find the leaf with smallest index $x_i$ and all non-leaves $x_j$ that may be connected to $x_i$ by an edge in the optimal tree (similarly to Algorithm~\ref{alg:degree_sequence_tree}).
For each choice of~$x_j$ we compute the optimal tree by adding the connecting edge $(x_j,x_i)$ to the optimal tree over $V \setminus \{x_i\}$, with the degree sequence suitably updated. 

The correctness of the dynamic programming algorithm follows from an observation similar to Lemma~\ref{lem2}; every subtree of the optimal tree must be optimal for its degree sequence, as otherwise it could be swapped for a cheaper subtree.
 The details are shown in Algorithm~\ref{alg:dpmv}.
\begin{algorithm}
  \caption{{\sc dp-MV} for generating the optimal directed tree for a given degree sequence.\label{alg:dpmv}}
  \begin{algorithmic}[1]
    \Statex \textbf{Input:} Degree sequence $(\delta^\text{out}, \delta^\text{in})$. 
    \Statex \textbf{Output:} $(T,c)$, where $T$ is optimum tree with root $x_1$, degrees $\delta^\text{out}, \delta^\text{in}$ and cost $c$.
    \Procedure{$\T\left[{{\dout}}, \din\right]$}{}
    \If{$\sum \din(\cdot) = 1$ }
      \State Let $i$ be the index where $\delta^\text{in}(x_i)=1$.
      \State \textbf{return} ($\{(x_1,x_i)\}$, $d_{x_1,x_i}$)
    \Else \, {\color{gray} $(\sum \din(\cdot) > 1)$}
      \State Let $i$ be the first index where $\dout(x_i)=0$ and $\din(x_i)=1$.
      \For{every index $j \neq i$, such that $\delta^\text{out}(x_j) + \delta^\text{in}(x_j) \geq 2$} \hfill$\triangleright$ attach $x_i$ to parent $x_j$
        \State $\Bigl( {\delta^\text{out}}', {\delta^\text{in}}' \Bigr)  \leftarrow \Bigl( \delta^\text{out}, \delta^\text{in} \Bigr) $
        \State ${\delta^\text{out}}' (x_j) \leftarrow {\delta^\text{out}}' (x_j) - 1$ 
        \State  ${\delta^\text{in}}' (x_i) \leftarrow 0$ 
        \State $(T_j,c_j) \leftarrow  \T[{\delta^\text{out}}' ,{\delta^\text{in}}' ]$
        \State ${T_j} \leftarrow T_j \cup \{(x_j,x_i)\}$
        \State ${c_j} \leftarrow c_j + d_{x_j,x_i}$
      \EndFor
      \State \textbf{return} $(T_j, c_j)$ with minimum $c_j$

    \EndIf
    \EndProcedure
  \end{algorithmic}
\end{algorithm}

\paragraph*{Analysis of Algorithm~\ref{alg:dpmv}.} Observe that the number of possible entries $\T[\cdot,\cdot]$, i.e.\ the number of feasible degree sequences for trees on subsets of $V$, is at most $\sum_{k=1}^{n} {n \choose k} {\DS(k)}$.
Using our previous estimate $\DS(n) = O(4^n)$, this sum evaluates, by the binomial theorem, to $O(5^n)$.
The overall run time of Algorithm~\ref{alg:outline}, including the transportation subproblem, with Algorithm~\ref{alg:dpmv} as a subroutine is $O^{\ast}(5^n)$, since the entry for a given degree sequence is computed at most once over all iterations, and the values are stored through the entire execution.
In practice, storing an entire tree in each $\T[\cdot,\cdot]$ is wasteful; for the optimum tree to be constructible from the table, it is sufficient to store the node to which the lowest-index leaf is connected. The claimed time and space complexities follow.

\subsection{{\sc dc-MV}: polynomial space and single-exponential time}
\label{secdc}
In this section we show how to reduce the space complexity to polynomial, while maintaining a single-exponential run time.

The outer loop (Algorithm~\ref{alg:outline}) remains the same, but we replace the subroutine for finding an optimal directed spanning tree (Algorithm~\ref{alg:dpmv}) with an approach based on divide and conquer (Algorithm~\ref{alg:dcmv}).
The approach is inspired by the algorithm of Gurevich and Shelah~\cite{GurevichShelah1987} for finding an optimal TSP tour.

The algorithm relies on the following observation about tree-separators.
Let $(V_1, V_2)$ be a partition of $V$, i.e.\ $V_1 \cup V_2 = V$, and $V_1 \cap V_2 = \emptyset$, and let $|V|= n$.
We say that $(V_1, V_2)$ is \emph{balanced} if $\max{ \bigl\{ |V_1|, |V_2| \bigr\} } \leq \ceil{2n/3}$. The following result is folklore (see e.g.\ ~\cite[\S\,9.1]{nishizeki1988planar}), we give a short proof for completeness.
\begin{lemma}
\label{lemsep}
  Every tree $T$ with vertex set $V$ admits a balanced partition $(V_1,V_2)$ of $V$ such that all edges of $T$ between $V_1$ and $V_2$ are incident to a vertex $v \in V_1$.
\end{lemma}
\begin{proof}
  A very old result of Jordan~\cite{Jordan1869} states that every tree has a \emph{centroid}, which is a vertex whose removal splits the tree into subtrees not larger than half of the original tree.
  To find such a centroid, move from an arbitrary vertex, one edge at a time, towards the largest subtree.

  Let $T_1, \dots, T_m$ be the vertex sets of the trees, in decreasing order of size, resulting from deleting the centroid of $T$. 
  Let $e_1, \dots, e_m$ denote the edges that connect the respective trees to the centroid in $T$.
  If $|T_1| \geq \floor{n/3}$, then we have the balanced partition $(V \setminus T_1, T_1)$, with a single crossing edge.

  Otherwise, let $m'$ be the smallest index such that $\floor{n/3} \leq |T_1| + \dots + |T_{m'}| \leq \ceil{2n/3}$.
  Observe that as $|T_i| < \floor{n/3}$ for all $i$, such an index $m'$ must exist.
  Now the balanced partition is\linebreak $(V \setminus (T_1 \cup \dots \cup T_{m'}), T_1 \cup \dots \cup T_{m'})$.
  The partition fulfils the stated condition as all crossing edges are incident to the centroid, which is on the left side.
\end{proof}

At a high-level, {\sc dc-MV} works as follows.
It \emph{guesses} the partition $(V_1, V_2)$ of the vertex set~$V$ according to a balanced separator of the (unknown) optimal rooted tree $T$, satisfying the conditions of Lemma~\ref{lemsep}.
It also \emph{guesses} the distinguished vertex $v \in V_1$ to which all edges that cross the partition are incident. (Throughout the description of the algorithm, guessing should be understood as iterating through all possible choices.)

The balanced separator splits $T$ into a tree with vertex set $V_1$ and a forest with vertex set~$V_2$. There are two cases to consider, depending on whether the root $r=x_1$ of $T$ falls in $V_1$ or $V_2$. 
In the first case, $V_1$ induces a directed subtree of $T$ rooted at $r$, in the second case, $V_1$ induces a directed subtree of $T$ rooted at $v$.
 
Observe that the out-degrees and in-degrees of vertices in $V_1$ are feasible for a directed tree, except for vertex $v$ which has additional degrees due to the edges crossing the partition. We refer to these out-degrees and in-degrees as the \emph{excess} of $v$.
The excess of $v$ can be computed using Lemma~\ref{lem:deg_seq_dir}.
This excess determines the number and orientation of edges across the cut, even if the endpoints, other than $v$, of the edges are unknown.  

By the same argument as in Lemma~\ref{lem2}, the subtree of $T$ induced by $V_1$ is optimal for its corresponding degree sequence after subtracting the excess of $v$, therefore we can find it by a recursive call to the procedure.

On the other side of the partition we have a \emph{collection} of trees.
To obtain an instance of our original problem, we add a virtual vertex $w$ to $V_2$ to play the role of $v$.
For all vertices $u$ we set the distance between $u$ and the virtual vertex $w$ to be the same as the distance between $u$ and $v$.
We set the out-degree and in-degree of $w$ to the \emph{excess} of $v$, to allow it to connect to all vertices of $V_2$ that~$v$ connects to in $T$. Now we can find the optimal tree on this side too, by a recursive call to the procedure.

\begin{algorithm}
  \caption{{\sc dc-MV} for generating the optimal directed tree for a given degree sequence.\label{alg:dcmv}}
  \begin{algorithmic}[1]
    \Statex \textbf{Input:} Vertex set $V$, degree sequence $\left(\dout, \din\right)$, and $n = |V|$.
    \Statex \textbf{Output:} $(T,c)$, where $T$ is optimum tree for $\delta^\text{out}, \delta^\text{in}$ and $c$ is its cost. 
    \Procedure{$\T\left[ V, \dout, \din\right]$}{}
    \If {$|V| \leq 5$}
	  \State Directly find optimum tree $T$ with cost $c$.
	  \State \textbf{return} ($T$, $c$)
    \Else
      \For {each partition $(V_1, V_2)$ of $V$, such that $|V_1|, |V_2| \leq \ceil{2n/3}$}        
        \For {each $v \in V_1$}
          \State $\Bigl( \mathsf{excess}^{\text{out}}_{v}, \mathsf{excess}^{\text{in}}_{v} \Bigr) \leftarrow \Bigl( \displaystyle\sum_{{u \in V_1}} \dout(u) - |V_1| + 1, ~~\displaystyle\sum_{{u \in V_1}} \din(u) - |V_1| + 1 \Bigr)$
          \State $\Bigl( {\delta^\text{out}}', {\delta^\text{in}}' \Bigr) \leftarrow \Bigl( \delta^\text{out}, \delta^\text{in} \Bigr)$
          \State ${\delta^\text{out}}' (v) \leftarrow {\delta^\text{out}}' (v) - \mathsf{excess}^{\text{out}}_{v}$, \quad ${\delta^\text{in}}' (v) \leftarrow {\delta^\text{in}}' (v) - \mathsf{excess}^{\text{in}}_{v}$
          \State $V'_2 \leftarrow V_2 \cup \{w\}$
		   \State $d_{wu} \leftarrow d_{vu}$, ~~$d_{uw} \leftarrow d_{uv}$,~~ for all $u \in V_2$

          \State $\Bigl( {\delta^\text{out}}' (w), {\delta^\text{in}}' (w) \Bigr) \leftarrow \Bigl( \mathsf{excess}^{\text{out}}_{v}, \mathsf{excess}^{\text{in}}_{v} \Bigr)$
          \If {${\delta^\text{out}}' (\cdot), {\delta^\text{in}}' (\cdot) \geq 0$ and ${\delta^\text{out}}' (v) + {\delta^\text{in}}' (v) \geq 1$ and ${\delta^\text{out}}' (w) + {\delta^\text{in}}' (w) \geq 1$ } 
            \State {$c_1 \leftarrow \T\left[V_1, ~~ {\delta^\text{out}}'  \rvert_{V_1}, ~~{\delta^\text{in}}'  \rvert_{V_1} \right]$} 
            \State {$c_2 \leftarrow \T\left[V'_2, ~~ {\delta^\text{out}}'  \rvert_{V'_2}, ~~{\delta^\text{in}}'  \rvert_{V'_2} \right]$} 
		    \State $c \leftarrow c_1 + c_2$
		    \State $T \leftarrow$ merge $T_1$ and $T_2$ by identifying $v$ and $w$
          \EndIf
        \EndFor
      \EndFor
    \State \textbf{return} $(T,c)$ with minimum $c$
    \EndIf
    \EndProcedure
  \end{algorithmic}
\end{algorithm}

On both sides of the partition, the roots of the directed trees are uniquely determined by the remaining degrees.
Observe that if the original root coincides with the centroid $v$, then $v$ and $w$ will be the roots of the trees in the recursive calls. After obtaining the optimal trees on the two sides (assuming the guesses were correct), we reconstruct $T$ by gluing the two trees together, identifying the vertices $v$ and $w$.
As this operation adds all degrees, we obtain a valid tree for the original degree sequence, furthermore, the tree must be optimal. 
We illustrate the two cases of this process in Fig.~\ref{fig:centroid1} with pseudocode given in Algorithm~\ref{alg:dcmv}. 

\begin{figure}
  \centering
  \includegraphics[width=12cm]{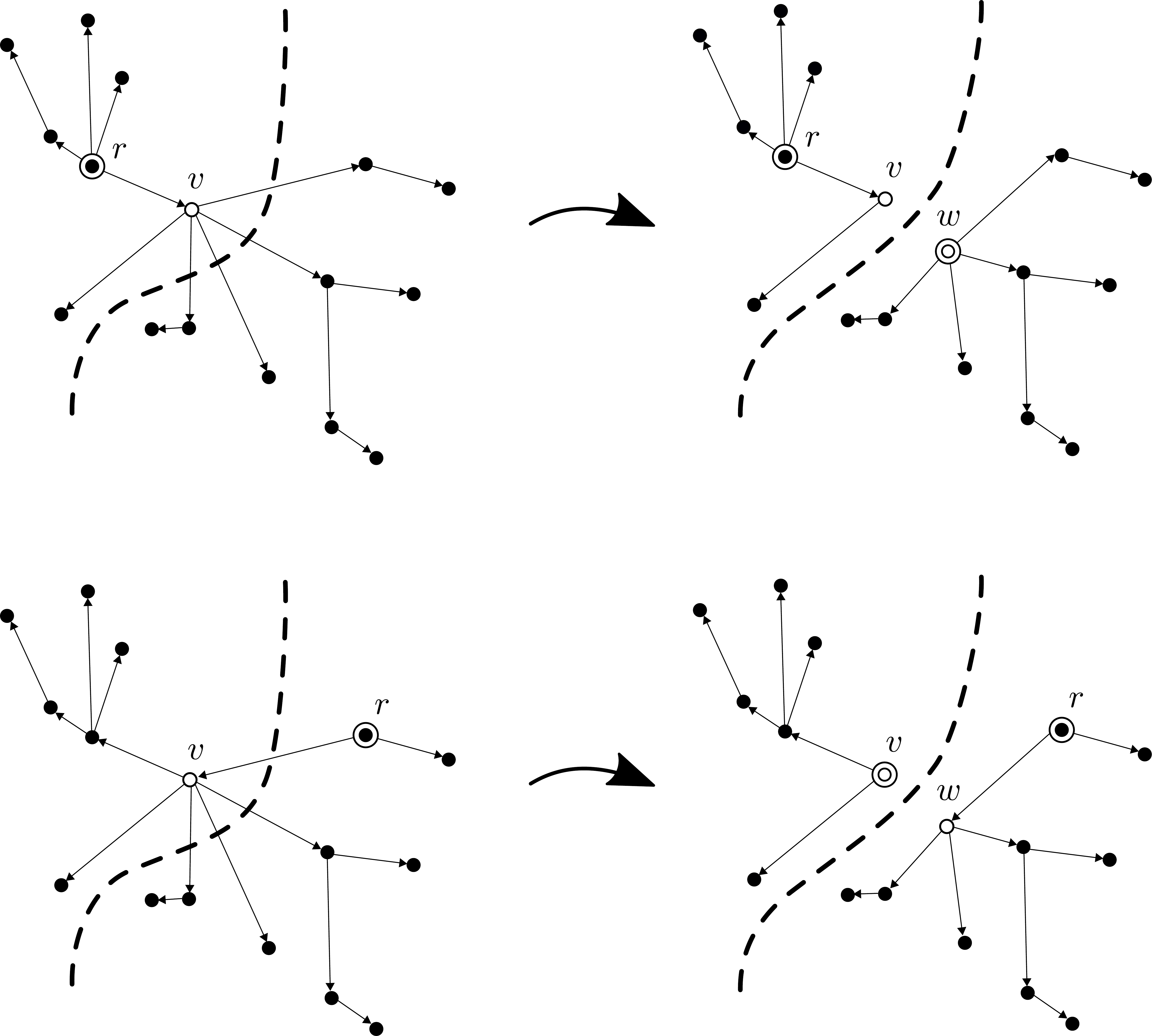}
  \caption{Illustration of {\sc dc-MV}. (left) optimal tree and balanced centroid-partitioning, the centroid shown as a hollow circle; (right) optimal trees on the two sides of the partitioning, centroid $v$ and its virtual pair $w$ shown as hollow circles; (above) root is in the left side of the partition, shown as double circle; (below) root is in the right side of the partition, shown as double circle.\label{fig:centroid1}}
\end{figure}

The threshold value $5$ in Line~2 of Algorithm~\ref{alg:dcmv} is to ensure progress; the subproblem size $\ceil{2n/3} + 1$ (the second term accounts for the virtual vertex $w$) needs to be less than $n$, which holds for $n>5$.
For $n \leq 5$, we proceed by a brute-force approach, constructing all feasible trees and choosing one with minimal cost.
In Line~8, the excess out-degree and in-degree of $v$ are calculated.
Both types of degrees need to sum to $|V_1|-1$, from which the expression follows. 
The condition in Line~14 ensures that we only solve feasible problems.
The notation $\delta \rvert_X$ indicates a restriction of a degree sequence to a subset $X$ of the vertices.

\paragraph*{Analysis of Algorithm~\ref{alg:dcmv}.} 
For a set $V$ of size $n$ we consider at most $2^n$ partitions (Line~6) and at most $\ceil{2n/3}$ choices of $v$ (Line~7), and we recur on subsets of size at most $\ceil{2n/3} + 1$.
All remaining operations take $O(n)$ time.
The run time $t(n)$ of {\sc dc-MV} thus satisfies: 
\begin{eqnarray*}
  t(n) & \leq & 2^n  \cdot n \cdot 2 t(\ceil{2n/3} + 1) \\
          &   =  & n^{O(\log n)} \cdot 2^{n \left( \sum_k{(2/3)^k }\right) + O(\log{n})}  \\
          &   =  & n^{O(\log n)} \cdot O(2^{3n}) \enspace .
\end{eqnarray*}
Including the transportation problem, the overall run time of Algorithm~\ref{alg:outline} using {\sc dc-MV} is thus $\DS(n) \cdot (8^n n^{O(\log n)} + n^3 \log n) + O(n^3 \log k) = O^*(32^{n+o(n)} + \log k)$.
For the space complexity, observe that we do not precompute $\T[\cdot,\cdot,\cdot]$ and, apart from minor bookkeeping, only a single tour is stored at each time, spread over $O(\log{n})$ recursive levels. The claimed bounds follow.

\subsection{{\sc dc-MV2}: polynomial space and improved single-exponential time}
\label{secdc2}

Finally, we modify the divide and conquer algorithm, improving its run time for a given degree sequence, from $8^{n+o(n)}$ to $4^{n+o(n)}$, while maintaining polynomial space.

The result is based on a strengthening of the structural observation about tree-separators (Lemma~\ref{lemsep}).
Specifically, we observe that the vertex set of a tree can be partitioned in a ``perfectly balanced'' way, if we allow the edges across the partition to be incident to a \emph{logarithmic} number of vertices on one side (instead of a single vertex as in Lemma~\ref{lemsep}). 

Again, let $(V_1, V_2)$ be a partition of $V(T)$, and let $|V(T)|= n$.
We say that $(V_1, V_2)$ is \emph{perfectly balanced} if $\max \bigl\{ |V_1|, |V_2| \bigr\} \leq \ceil{n/2}$. 
\begin{lemma}
\label{lemsep2}
  Every tree $T$ admits a perfectly balanced partition $(V_1,V_2)$ such that all edges of $T$ between~$V_1$ and $V_2$ are incident to a vertex in $\{v_1, \dots, v_k\} \subseteq V_1$, where $k \leq \floor{\log_2 n}$.
\end{lemma}

\begin{figure}[h]
  \centering
  \includegraphics[width=6cm]{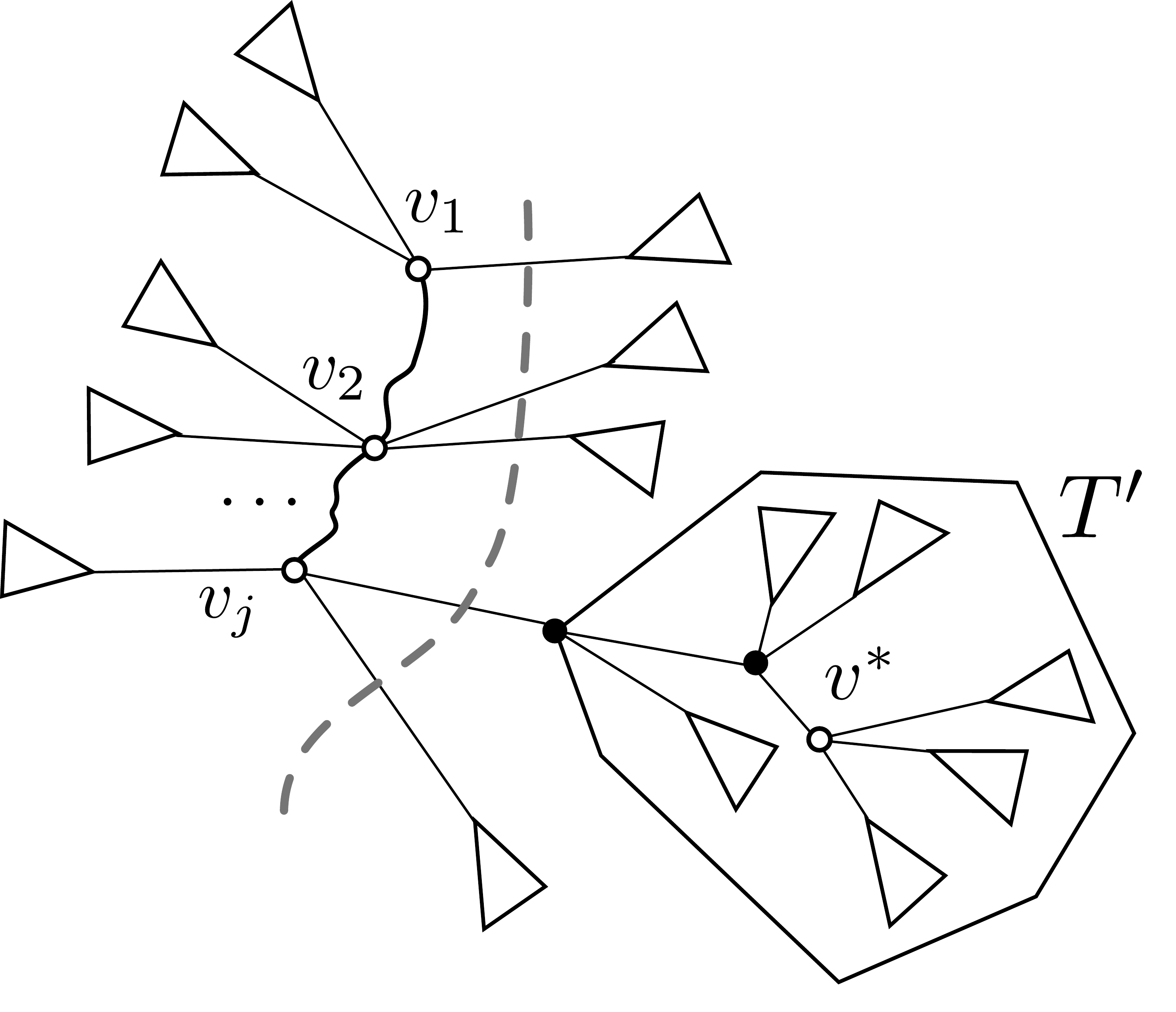}
  \caption{Illustration of the proof of Lemma~\ref{lemsep2}. \label{fig:ind}}
\end{figure}

\begin{proof}
  We proceed by finding the special vertices $v_j$ one-by-one for $j = 1, \dots, \floor{\log_2 n}$ and updating the partition $(V_1, V_2)$ in each iteration, maintaining the property that all edges between $V_1$ and $V_2$ are incident to a vertex in $\{v_1, \dots, v_j\}$. 
  We additionally require the following conditions in every iteration: (1) $v_1,\dots,v_j$ are, in this order, on a simple path entirely within $V_1$, and (2) $|V_1| \leq |V_2|$.
  Furthermore, if $|V_1| < |V_2|$, then we require that there is a subtree $T'$ of $T$ attached to~$v_j$ such that (3)~$T'$ is entirely contained in $V_2$, (4) $|T'| \leq n / {2^j}$, and (5) $|V_1| + |T'| > n/2$.
  In words, moving~$T'$ to the left side of the partition would make the left side larger than the right side. 
  The claim of the Lemma follows from conditions (1)--(5), when $j= \floor{\log_2 n}$.

  Let us prove that the conditions hold by induction on $j$.
  For the base case, $j=1$, we reuse our construction from Lemma~\ref{lemsep}.
  Let $v_1$ be the centroid of $T$, and iteratively move the subtrees of $T$ attached to $v_1$ to the initially empty $V_2$ from the initially full $V_1$, until $|V_2| \geq |V_1|$.
  Denote the last subtree that is moved as $T'$.
  Conditions (1)--(5) are easily verified. 

  For the inductive step, assume that $v_1, \dots, v_j$ have been found, satisfying conditions (1)--(5), and let $T'$ be the subtree attached to $v_j$ that fulfills conditions (3)--(5).
  Let $v^\ast$ be the centroid of~$T'$, and let $q_1, \dots, q_m$ be the internal vertices (in this order) on the path from $v_j$ to~$v^\ast$.
  Let $q_{m+1}=v^\ast$.
  Iteratively, for $\ell=1,\dots,m+1$, move the nodes $q_\ell$ from $V_2$ to $V_1$, and after moving each $q_\ell$, move one-by-one the subtrees of $T'$ attached to $q_\ell$, other than the one containing the centroid~$v^\ast$.
  By the definition of the centroid, and by (4), all subtrees that we move have size at most $n / {2^{j+1}}$.
  We stop when $|V_1| \geq |V_2|$. Observe that this must happen eventually, due to condition (5). 

  If the last move is a vertex $q_\ell$, then $|V_2| \leq |V_1| \leq |V_2| + 1$, so a perfectly balanced partition has been found; we label $q_\ell$ as $v_{j+1}$ and stop the entire process. In this case, verifying conditions (1)--(5) is no longer needed. 
  If the last move is a subtree, then we move it back to $V_2$, label it $T'$, and label the last moved vertex $q_\ell$ as $v_{j+1}$. 
  By construction, all edges across the partition are incident to one of the vertices $v_1, \dots, v_{j+1}$. This is because intermediate vertices $q_1, \dots, q_{\ell-1}$ have been moved across the partition together with all their pendent trees, they cannot thus be incident to any edges across the partition.
  (Recall that $v_1, \ldots, v_{j+1}$ form a simple path, but are not necessarily consecutive vertices on this path.)
     By construction, conditions (1)--(5) are satisfied with $v_{j+1}$, the new $T'$, and the current $V_1$, $V_2$. 
  See Fig.~\ref{fig:ind} for illustration. 
\end{proof}

Our improved algorithm {\sc dc-MV2} follows the same outline as {\sc dc-MV}, described in \S\,\ref{secdc}. We highlight the main differences.

The algorithm \emph{guesses} the partition $(V_1,V_2)$ of the vertex set $V$, according to a perfectly balanced separator of the (unknown) optimal rooted tree $T$, satisfying the conditions of Lemma~\ref{lemsep2}.
It also guesses the distinguished vertices $\{v_1, \dots, v_k\} \subseteq V_1$ to which all edges that cross the partition are incident.
Again, the separator splits $T$ into a tree with vertex set $V_1$ and a forest with vertex set~$V_2$.
Again, we consider separately the cases when the root $r = x_1$ of $T$ falls in $V_1$ or $V_2$.
In the first case,~$V_1$ induces a directed subtree of $T$ rooted at $r$, in the second case, $V_1$ induces a directed subtree of $T$ rooted at $v_i$, for some $1 \leq i \leq k$. (See Fig.\ \ref{fig:centroid2}.)

Whereas in {\sc dc-MV} all excess degrees belong to a single vertex, they are now distributed among $v_1, \dots, v_k$.
Therefore, we compute the \emph{total excess} of the vertices $v_1, \dots, v_k$ and we \emph{guess} their distribution.
Distributing the excess \emph{in-degree} is easy.
Assuming that our guesses were correct, at most one of the vertices $v_1,\dots,v_k$ has an incoming edge across the partition.
We just need to guess therefore, which vertex $v_i$ (if any) has an excess in-degree of $1$, while all others have excess in-degree $0$ (see Fig.\ \ref{fig:centroid2}).
For the excess \emph{out-degrees}, we need to distribute a total value of at most $n$ among $k = O(\log{n})$ vertices, amounting to $n^{O(\log{n})}$ possible choices.

We update the degrees of $v_1, \dots, v_k$ in $V_1$, removing their excesses.
Assuming our guesses to be correct, the subtree of $T$ induced by $V_1$ is optimal for its degree sequence, by Lemma~\ref{lem3}.

\begin{algorithm}[h]
  \caption{{\sc dc-MV2} for generating the optimal directed tree for a given degree sequence.\label{alg:best}}
  \begin{algorithmic}[1]
    \Statex \textbf{Input:} Vertex set $V$, degree sequence $\left(\dout, \din\right)$, and $n = |V|$.
    \Statex \textbf{Output:} $(T,c)$, where $T$ is optimum tree for $\delta^\text{out}, \delta^\text{in}$ and $c$ is its cost. 
    \Procedure{$\T\left[ V, \dout, \din\right]$}{}
    \If {$|V| \leq 9$}
	  \State Directly find optimum tree $T$ with cost $c$.
	  \State \textbf{return} ($T$, $c$)
    \Else
      \For {each partition $(V_1, V_2)$ of $V$, such that $|V_1|, |V_2| \leq \ceil{n/2}$}        
        \For {each $\{v_1, \dots, v_k\} \subseteq V_1$, where $k \leq \floor{\log_2 n}$}
            \State $\Bigl( \mathsf{excess}^{\text{out}}, \mathsf{excess}^{\text{in}} \Bigr) \leftarrow \Bigl( \displaystyle\sum_{{u \in V_1}} \dout(u) - |V_1| + 1, ~~\displaystyle\sum_{{u \in V_1}} \din(u) - |V_1| + 1 \Bigr)$
            \For {each \emph{excess-partition} $\Bigl\{ \mathsf{excess}^{\text{out}}_{v_i} \Bigr\}_{1 \leq i \leq k}$,~~ $\Bigl\{ \mathsf{excess}^{\text{in}}_{v_i} \Bigr\}_{1 \leq i \leq k}$}
              \State $\Bigl( {\delta^\text{out}}', {\delta^\text{in}}' \Bigr) \leftarrow \Bigl( \delta^\text{out}, \delta^\text{in} \Bigr)$        
              \State  Update degrees and distances appropriately. \hfill   $\triangleright$ Algorithm~\ref{alg:update}        
              \If {${\delta^\text{out}}', {\delta^\text{in}}'$ are \emph{valid} for directed trees } 
                \State {$c_1 \leftarrow \T\left[V_1, ~~ {\delta^\text{out}}'  \rvert_{V_1}, ~~{\delta^\text{in}}'  \rvert_{V_1} \right]$} 
                \State {$c_2 \leftarrow \T\left[V'_2, ~~ {\delta^\text{out}}'  \rvert_{V'_2}, ~~{\delta^\text{in}}'  \rvert_{V'_2} \right]$} 
		        \State $c \leftarrow c_1 + c_2$
		        \State $T \leftarrow$ merge $T_1$ and $T_2$ by identifying $v_i$ and $w_i$ and removing $w_0$
              \EndIf
            \EndFor
          \EndFor
        \EndFor
      \State \textbf{return} $(T,c)$ with minimum $c$
    \EndIf
  \EndProcedure
  \end{algorithmic}
\end{algorithm}

On the side of $V_2$ we have a collection of trees.
To obtain an instance of our problem (i.e.\ a single tree), we add virtual vertices $w_1, \dots, w_k$ to $V_2$, to play the role of $v_1, \dots, v_k$, and a virtual root~$w_0$ forcibly connected to the virtual vertices. 

To achieve the correct connections, we first set the out-degrees and in-degrees of $w_1, \dots, w_k$ to match the excesses of $v_1, \dots, v_k$.
This allows them to connect to all vertices of $V_2$ that $v_1,\dots,v_k$ connect to in $T$.
Observe that the location of the root of $T$ is determined by the in-degrees of the virtual vertices.
If all $w_i$ have in-degree $0$, then the root is on the left side of the partition (i.e.\ in $V_1$).
If $w_\ell$, for some $\ell\in\{1,\dots,k\}$ has in-degree $1$, then the root is on the right side of the partition (i.e.\ in $V_2$), in a subtree attached to $v_\ell$. 

In the first case, we make sure that all edges $(w_0,w_i)$ are directed away from~$w_0$.
In the second case, edge $(w_\ell, w_0)$ is directed towards $w_0$, and all edges $(w_0, w_i)$ for $i= 1,\dots,k$, and $i \neq \ell$ are directed away from $w_0$ (see Fig.\ \ref{fig:centroid2}).
To enforce these connections, we appropriately adjust the in-degrees and out-degrees of $w_0, \dots, w_k$.
We also set distances involving virtual vertices $w_0, \dots, w_k$, such as to forbid unwanted connections. (See Algorithm~\ref{alg:update}.)

After obtaining the optimal trees on the two sides (assuming the guesses were correct), we reconstruct $T$ by gluing the two trees together, removing $w_0$ with all its incident edges, and identifying $v_i$ with $w_i$ for all $i = 1,\dots,k$.
The optimality of the obtained tree follows by observing that a tree with smaller cost would also induce an improved solution on at least one of the subproblems, contradicting their optimality.
Again, we remark that ``guessing'' means iterating through all possible choices.
We refer to the detailed pseudocode described in Algorithm~\ref{alg:best} and the illustration in Fig.\ \ref{fig:centroid2}.
\begin{figure}[h!]
	\centering
	\includegraphics[width=11cm]{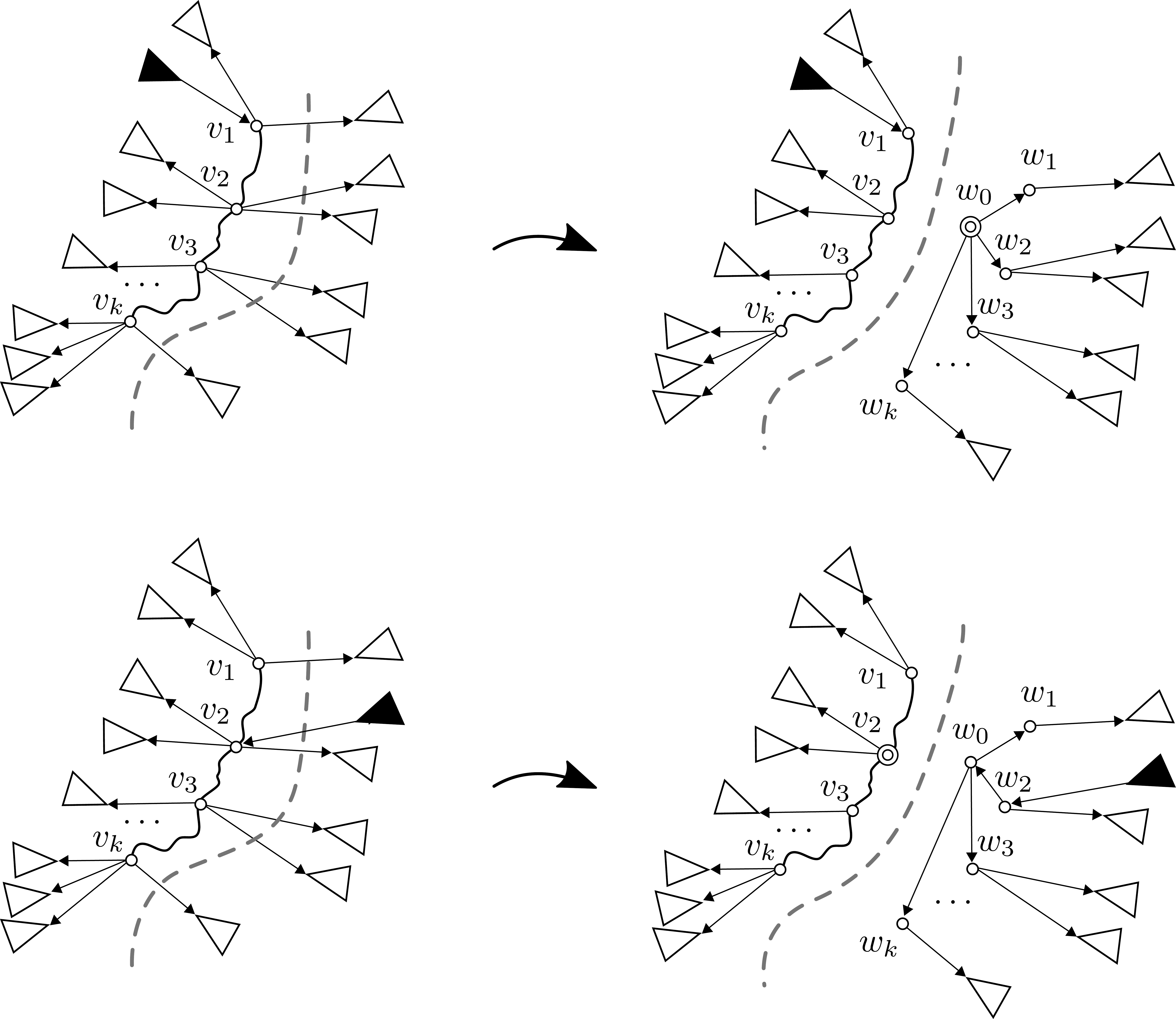}
	\caption{Illustration of {\sc dc-MV2}. (left) optimal tree and perfectly balanced partitioning; triangles indicate subtrees, shaded triangle contains root. (right) optimal trees on the two sides of the partitioning, virtual vertices $w_0,w_1,\dots,w_k$ shown; (above) root falls to the left side of the partition; (below) root falls to the right side of the partition; new roots shown as double circles.\label{fig:centroid2}}
\end{figure}

Some further remarks about Algorithm~\ref{alg:best} are in order.
The threshold 9 in Line~2 is to ensure progress; we need the recursive subproblem size $\ceil{n/2} + \floor{\log_2 n} + 1$ (the additive logarithmic term is due to the virtual vertices) to be less than $n$, which holds for $n>9$.
In practice, even some values $n \leq 9$ may be too large to be solved by brute force; in such cases we can switch back to Algorithm~\ref{alg:dcmv} (\textsc{dc-MV}) or Algorithm~\ref{alg:dpmv} (\textsc{dp-MV}), without affecting the asymptotic guarantees.

Lines $6,7$, and $9$ contain the ``guesses'', specifically, for the partition of $V$, for the choice of $v_1, \dots, v_k$, and for distributing the excess degrees among $v_1, \dots, v_k$.
Line 12 requires a test similar to the corresponding test in Algorithm~\ref{alg:dcmv}, that the guessed values correctly describe a tree on both sides of the partition. Line 11 contains the update of degrees and distances on the two sides of the partition, separated into Algorithm~\ref{alg:update} for readability. 

\paragraph*{Analysis of Algorithm~\ref{alg:best}.}
For a set $V$ of size $n$ we consider at most $2^n$ partitions (Line~6), at most $n^{O(\log{n})}$ choices for the $v_i$ (Line~7), and at most $n^{O(\log{n})}$ choices for their excess in-degrees and out-degrees (Line~9). 
We recur on subsets of size at most $\ceil{n/2} + \floor{\log_2 n} + 1$.

All remaining operations take $O(n)$ time.
The run time $t(n)$, can thus be bounded as: 
\begin{eqnarray*}
  t(n) & \leq & 2^n  \cdot n^{O({\log{n}})}\cdot 2 t(\ceil{n/2} + \floor{\log_2 n}+1) \\
          &   =  & n^{O(\log^2 n)} \cdot 2^{n \left( \sum_k{(1/2)^k}\right)+ O(\log^2{n})}   \\
          &   =  & n^{O(\log^2 n)} \cdot O(2^{2n}) \enspace .
\end{eqnarray*}
Including the transportation problem, the overall run time of Algorithm~\ref{alg:outline} using {\sc dc-MV2} is therefore $\DS(n) \cdot (4^n n^{O(\log n)} + n^3 \log n) + O(n^3 \log k) = O^*(16^{n+o(n)} + \log k)$, with essentially the same space complexity as before. 

We remark that in instances of MV-TSP where some of the multiplicities are smaller than $n-1$, not all degree sequences of trees are realizable.
Pruning out unrealizable degree sequences can lead to improvements in efficiency.
As a special case, suppose that all multiplicities are equal to $1$ (i.e.\ the standard TSP).
Then, the only possible tree is a \emph{path}, with a unique degree sequence, up to the choice of the leaf.
In this case the run time of our approach reduces to $4^{n + o(n)}$, matching the Gurevich-Shelah algorithm~\cite{GurevichShelah1987}.

\section{Discussion}
\label{sec:discussion}
We described three new algorithms, along with variations, for the many-visits TSP problem.
In particular, we showed how the problem can be solved in time single-exponential in the number~$n$ of cities, while using space only polynomial in $n$.
This yields the first improvement over the Cosmadakis-Papadimitriou algorithm in more than 35 years.

It remains an interesting open question to improve the bases of the exponentials in our run times, even in the special case when all edge costs are equal to $1$ or $2$.
Such instances of MV-TSP arise e.g.\ in the {\sc Maximum Scatter TSP} application by Kozma and M{\"o}mke~\cite{KozmaMomke2017}. Recent algebraic techniques~\cite{Algebraization, Golovnev, Koivisto} may be of help.


In practice, one may reduce the search space of our algorithms via heuristics, for instance by forcing certain (directed) edges to be part of the solution.
An edge $(i,j)$ is part of the solution if the optimal tour visits it \emph{at least once}.
This may be reasonable if the edges in question are very cheap.

\begin{algorithm}[h]
  \caption{Updating degrees and distances in {\sc dc-MV2}: Line 12 of Algorithm~\ref{alg:best}.\label{alg:update}}
    \begin{algorithmic}[1]

 \State ${\delta^\text{out}}' (v_i) \leftarrow {\delta^\text{out}}' (v_i) - \mathsf{excess}^{\text{out}}_{v_i}$, \ \  ${\delta^\text{in}}' (v_i) \leftarrow {\delta^\text{in}}' (v_i) - \mathsf{excess}^{\text{in}}_{v_i}$, \ \ for $1 \leq i \leq k$
          \State $V'_2 \leftarrow V_2 \cup \{w_0,w_1,\dots,w_k\}$ \hfill$\triangleright$ virtual vertices

 \State ${\delta^\text{out}}' (w_i) \leftarrow  \mathsf{excess}^{\text{out}}_{v_i}$, \ \  ${\delta^\text{in}}' (w_i) \leftarrow \mathsf{excess}^{\text{in}}_{v_i}+1$, \ \  for all $1 \leq i \leq k$
		   \State $d_{w_0, w_i} \leftarrow 0$, \ \ $d_{w_i, w_0} \leftarrow 0$, \ \  for all $1 \leq i \leq k$
   		   \State $d_{w_0, u} \leftarrow \infty$, \ \ $d_{u,w_0} \leftarrow \infty$, \ \ for all $u \in V_2 \setminus \{w_1,\dots,w_k\}$
		   \State $d_{w_i, w_j} \leftarrow \infty$, \ \ for all $1 \leq i,j \leq k$
		   \State $d_{w_i, u} \leftarrow d_{v_i, u}$,\ \ and \  $d_{u, w_i} \leftarrow d_{u, v_i}$ \ for all $1 \leq i \leq k$, \ and \ $u \in V_2 \setminus \{w_0,\dots,w_k\}$
		   
          \State ${\delta^\text{out}}' (w_0) \leftarrow k$, \ \ ${\delta^\text{in}}' (w_0) \leftarrow 0$

   \If {$\mathsf{excess}^{\text{in}}_{v_i} > 0$ for some $i\in\{1,\dots,k\}$} \hfill$\triangleright$ root is in $V_2$, in subtree attached to $v_i$ 

	\State ${\delta^\text{out}}' (w_i) \leftarrow {\delta^\text{out}}' (w_i) + 1$
	\State ${\delta^\text{in}}' (w_i) \leftarrow {\delta^\text{in}}' (w_i) - 1$
	\State ${\delta^\text{out}}' (w_0) \leftarrow {\delta^\text{out}}' (w_0) - 1$
	\State ${\delta^\text{in}}' (w_0) \leftarrow {\delta^\text{in}}' (w_0) + 1$

   \EndIf

 \end{algorithmic}
\end{algorithm}

\begin{acks}
  Research of L.K.\ supported by ERC Consolidator Grant No 617951 and DFG Grant KO 6140/1-1. Research of M.M.\ supported by DFG Grant MN 59/4-1. We thank the anonymous reviewers for their many insightful remarks.
\end{acks}


\clearpage

\bibliographystyle{ACM-Reference-Format}
\bibliography{submission}

\newpage

\appendix

\section{Deferred subroutines}
\label{combalg}

\begin{algorithm}
  \caption{Generating all possible $r$-subsets of $\{1, \dots, n\}$.\label{alg:combinations}}
  \begin{algorithmic}[1]
    \State \textbf{Input:} Positive integers $n$ and $r$.
    \State \textbf{Output:} A \emph{generator} of all possible $r$-subsets of $\{1, \dots, n\}$.
    \Procedure{combinations}{$n$, $r$}
      \State $\mathsf{comb} \leftarrow [1, \dots, r]$
      \State \textbf{yield} $\mathsf{comb}$
      \While{\textbf{true}}
        \State{$i \leftarrow$ last index such that $\mathsf{comb}_i \neq n-r+i$, if no such index exists, \textbf{break}}
        \State{$\mathsf{comb}_i \leftarrow \mathsf{comb}_{i} + 1$}
        \For{every index $j \leftarrow i{+}1, \dots, r$}
          \State $\mathsf{comb}_j \leftarrow \mathsf{comb}_{j-1} + 1$
        \EndFor
        \State \textbf{yield} $\mathsf{comb}$
      \EndWhile
    \EndProcedure
  \end{algorithmic}
\end{algorithm}

The algorithm $\textsc{combinations}$ implements algorithm $\textsc{NEXKSB}$ by Nijenhuis and Wilf~\cite[page 26]{NijenhuisWilf1978}.
It takes two integers as input, $n$ and $r$, and generates all (ordered) subsets of size $r$, of the base set $\{1, \dots, n\}$.
It starts with the set $[1, 2, \dots, r]$ and generates all $r$-subsets in lexicographical order, up to $[n-r+1, n-r+2, \dots, n-r+r] = [n-r+1, n-r+2, \dots, n]$.
In every iteration, it increments the rightmost entry $\mathsf{comb}_i$ not equal to $n-r+i$, and makes $\mathsf{comb}_j$ equal to $\mathsf{comb}_{j-1}+1$ for all $j$ indices between $[i+1, r]$.
The algorithm stops when there are no such indices $i$, which happens when reaching $[n-r+1, \dots, n]$. 

An example is shown in Fig.~\ref{fig:combinations}.
The corresponding integer sequence would be $[1, 1, 2, 0, 0]$, however $\textsc{combinations}$ returned the sequence $[2, 4, 7, 8]$, that is, a sequence of the positions (separating bars) of the integer sequence above.
In order to obtain the actual degree sequence, we use a short script \textsc{combinationToSequence} in Algorithm~\ref{alg:combination_to_sequence} that converts the sequence with the positions of the bars to the degree sequence.

\begin{figure}[h]
  \centering
  \resizebox {.3\textwidth} {!} {
    \begin{tikzpicture}
      \tikzstyle{star}=[circle,draw,fill, minimum size=14pt, inner sep=0pt]
      \tikzstyle{bar}=[rectangle,draw,fill,rounded corners=0.5pt, minimum width=2pt, minimum height=20pt, inner sep=0pt]
      \node[star, label={[label distance=8pt] \small 1}] (p1) at (0,0) {};
      \node[bar,  label={[label distance=5pt] \small 2}, right=6pt of p1] (p2) {};
      \node[star, label={[label distance=8pt] \small 3}, right=6pt of p2] (p3) {};
      \node[bar,  label={[label distance=5pt] \small 4}, right=6pt of p3] (p4) {};
      \node[star, label={[label distance=8pt] \small 5}, right=6pt of p4] (p5) {};
      \node[star, label={[label distance=8pt] \small 6}, right=6pt of p5] (p6) {};
      \node[bar,  label={[label distance=5pt] \small 7}, right=6pt of p6] (p7) {};
      \node[bar,  label={[label distance=5pt] \small 8}, right=6pt of p7] (p8) {};
    \end{tikzpicture} }
    \caption{Sequence $[2, 4, 7, 8]$ representing the degree sequence $[1, 1, 2, 0, 0]$}
\label{fig:combinations}
\end{figure}
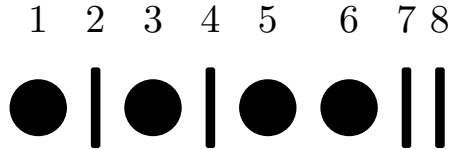

\begin{algorithm}[h]
  \caption{Converting $[a_1, \dots, a_m]$ into $[a_1, a_2-a_1-1, \dots, r+m-a_m]$.\label{alg:combination_to_sequence}}
  \begin{algorithmic}[1]
    \State \textbf{Input:} List of positions $a=[a_1, \dots, a_m]$, integer $r$.
    \State \textbf{Output:} A sequence of $m+1$ integers that sum up to $r$.
    \Procedure{combinationToSequence}{$a$, $r$}
      \State $\mathsf{seq}_1 \leftarrow a_1-1$
      \For{ every index $i \leftarrow 2,\dots,m$}
        \State $\mathsf{seq}_i \leftarrow a_{i}-a_{i-1}-1$
      \EndFor
      \State $\mathsf{seq}_{m+1} \leftarrow r+m-a_{m} $
      \State \textbf{return} $\mathsf{seq}$
    \EndProcedure
  \end{algorithmic}
\end{algorithm}

Finally, the procedure {\sc Distribute}$(n,k)$ calls {\sc CombinationToSequence}$(a,n)$, for each output $a$ of {\sc Combinations}$(n+k,n)$.

\end{document}